\documentclass[format=acmsmall, review=false]{acmart}

\usepackage{acm-ec-18}
\usepackage{booktabs} 
\usepackage[ruled]{algorithm2e} 
\usepackage{amsmath,amsfonts,amssymb,graphicx,indentfirst,amsthm,bm}
\usepackage[utf8]{inputenc}
\usepackage[normalem]{ulem}
\usepackage{dsfont}
\usepackage{myarticle-macrosShared}

\newcommand{\vValue}{\psi}
\newcommand{\indicator}[1]{\mathds{1}_{\left[{#1}\right]}}

\SetAlFnt{\small}
\SetAlCapFnt{\small}
\SetAlCapNameFnt{\small}
\SetAlCapHSkip{0pt}
\IncMargin{-\parindent}
\DeclareMathOperator{\Unif}{Unif}
\title[Explicit shading strategies for repeated truthful auctions]{Explicit shading strategies for repeated truthful auctions}
\author{Marc Abeille}
\affiliation{Criteo Research, m.abeille@criteo.com}
\author{Clement Calauzenes}
\affiliation{Criteo Research, c.calauzenes@criteo.com}
\author{Noureddine El Karoui}
\affiliation{Criteo Research and UC, Berkeley, n.elkaroui@criteo.com, nkaroui@berkeley.edu}
\author{Thomas Nedelec}
\affiliation{Criteo Research, ENS Paris Saclay, t.nedelec@criteo.com}
\author{Vianney Perchet}
\affiliation{ENS Paris Saclay, Criteo Research, perchet@normalesup.org}

\date{March 2018}

\begin{document}
\begin{abstract}
With the increasing use of auctions in online advertising, there has been a large effort to study seller revenue maximization, following Myerson's seminal work, both theoretically and practically.
We take the point of view of the buyer in classical auctions and ask the question of whether she has an incentive to shade her bid even in auctions that are reputed to be truthful, when  aware of the revenue optimization mechanism.

We show that in auctions such as the Myerson auction or a VCG with reserve price set as the monopoly price, the buyer who is aware of this information has indeed an incentive to shade. 
Intuitively, by selecting the revenue maximizing auction, the seller introduces a dependency on the buyers' distributions in the choice of the auction.
We study in depth the case of the Myerson auction and show that a symmetric equilibrium exists in which buyers shade non-linearly what would be their first price bid. They then end up with an expected payoff that is equal to what they would get in a first price auction with no reserve price. 

We conclude that a return to simple first price auctions with no reserve price or at least non-dynamic anonymous ones is desirable from the point of view of both buyers, sellers and increasing transparency. 

\end{abstract}

\maketitle

\section{Introduction}
\label{sec:introduction}
Billions of auctions are run worldwide everyday. One of the main supplier of such auctions is the online advertising market \cite{AllouahBesbes2017,Balseiro2015MultiStage,AmiKeaKey12}. Ad slots are sold to  advertisers by a publisher, typically a web site, following  more or less explicit mechanisms, i.e., a type of auctions with specific rules. Those auctions take place on platforms known as ``ad exchanges'' \cite{Mut09}.

Maybe the most common types of auctions used in this setting are the classical second price auctions with or without reserve prices as they are reputed to be truthful (it is dominant to bid the true valuation) and even optimal for identical bidders \cite{Myerson81,RilSam81}. Auction theory has been developed for several decades (and several Nobel prices were granted for those breakthroughs). As a consequence, optimal strategies of the bidders (when they exist) and revenue maximizing mechanisms are almost perfectly understood.

However, one of the crucial working assumption is that the seller must know the distribution of valuations of the bidders to tune optimally her mechanism (for instance by setting the right reserve price). In practice though, this assumption is obviously not satisfied \cite{Wil87}. On the other hand, in most real life applications, such as online advertising, auctions between a single publisher and the same advertisers are not run only once, but  several hundred thousand times each day. As a consequence, the seller has access to an incredibly large amount of bids from each bidder for more or less equivalent goods. So one trend of research is to learn the optimal mechanisms from the past sequence of bids, assuming that they truthfully represent the valuations of the bidders \cite{CesGenMan13,FuHarHoy13,OstSch11,ColRou14,AmiCumDwo15,KanNaz14,DhaRouYan15,BluManMor15,MohMed14,AmiRosSye14,MorgenRough2016,ChaHarKle07,DasSyr16}.

The motivation behind this traditional assumption is that the seller is only choosing incentive compatible auctions such as Vickrey auctions. Therefore, since in a one shot second price auctions it is optimal to bid one's own valuations, the seller can safely expect to observe the past valuations of the buyer, and hence an approximation of her distribution of valuations. Even if computing the optimal auction of Myerson might be complicated, as it needs an almost  number of samples, there exist approximatively optimal auctions \cite{HarRou09,FuHarHoy13,Hartline2009} requiring a much more reasonable amount of data.

\medskip

Our starting point is the following claim. If the seller uses data acquired on past auctions to update the  mechanism (say, to fix dynamically the reserve price) then the repeated mechanism might no longer be truthful. Intuitively, this is rather clear. Assume that the reserve price in the second price auction of a seller is determined by her past bids. By bidding untruthfully (even in Vickrey auctions!), the bidder might lower drastically her reserve price, at a small cost of lost auctions. As a result, her bids might actually clear the new reserve price more often than with truthful bids.

We actually prove this intuition that several mechanisms known to be truthful in the one shot case are no longer truthful in the repeated setting. Assuming symmetrical bidders and a seller that myopically designs a Myerson auction with respect to the past sequence of bids (in effect assuming they were equal to the valuations), we prove that a symmetric equilibrium enforce the same outcome than second price auctions without reserve prices !

Actually, our results do not rely on the finite sample setting that has received a lot of attention recently (both from the seller, but also on the bidder point of view \cite{McA11,Weed16}). We will directly assume that a strategy of a bidder is a modification (or a function) of her true valuation, that the seller observes the distribution of bids, from which she computes reserve prices. The payoff of a strategy (a.k.a a modifying or shading function) will therefore be computed as the expected gain of the auctions run with all those parameters. Although interesting, putting an extra layer of learning over datasets of finite sample size is not necessarily to convey our main messages.

The paper is organized as follows. We introduce the model in Section 2. In Section 3, we prove the claim that repeated ``one-shot truthful’’ auctions (for several widely used types of auction) are no longer truthful by exhibiting a simple strategy that dominates truthful bidding. We investigate in detail the Myerson auction in Section 4, where we prove that a symmetric equilibrium enforces the same outcome as second price auctions without reserve prices. The Appendix contains more technical results, reminders and proofs of some lemmas. 

\section{Setting}
\label{sec:setting}

We consider $K$ independent bidders with distributions of valuation $F \in \cF$ participating to repeated independent auctions:  the values $X_i$ of bidder $i$ in the different auctions are i.i.d., drawn from $F_i$. We consider only stationary strategies, i.e., the bidder $i$ applies a function $\beta_i: \lR \mapsto \lR$ mapping her value to a bid, and we denote by $\cB \subseteq \{\lR\mapsto\lR\}^K$ the set of strategy profiles. We will assume the $\beta_i$ to be increasing -- i.e. the higher the value, the higher the bid. We denote $H=(H_1 \dots H_K) \in \cH(\cF, \cB)$ the resulting distributions of bids -- i.e. $H_i$ is the distribution of $B_i = \beta_i(X_i)$ where $X_i$ is drawn from $F_i$. Especially, as the $\beta_i$ are increasing, we can have
$$
H_i(b) = F_i(\beta_i^{-1}(b)) ~~~~\text{and}~~~~ h_i(b) = \frac{f_i(\beta_i^{-1}(b))}{\beta_i'(\beta_i^{-1}(b))}
$$

Another quantity of interest in auction studies, only depending on the distribution of a r.v. $X$ with distribution $F$, is the virtual value of a bidder $\vValue_X(x) = x - \frac{1 - F(x)}{f(x)}$ and the related hazard rate $\lambda_X(x) = \frac{f(x)}{1 - F(x)}$. We can notice we have the following link between the virtual values of $X$ and $B$:
\begin{align}
   \vValue_B(\beta(x)) - \beta(x) = \beta'(x)(\vValue_X(x) - x)
   ~~~\text{and}~~~
   \lambda_B(\beta(x)) = \frac{\lambda_X(x)}{\beta'(x)}
   \label{eq:edp_beta_virtual_values}
\end{align}
Hence, up to solving the ODE~\eqref{eq:edp_beta_virtual_values}, we can manipulate equivalently the bidder's strategy $\beta$, the virtual value $\vValue_B \circ \beta$ at value $x$ exposed to the seller or the corresponding hazard rate $\lambda_B \circ \beta$.

\subsection{Revenue Maximizing Auctions}
With the widespread use of auctions in online advertising markets, there has been an extensive work about characterizing and estimating revenue-maximizing auctions in different settings. 
More recently, some focus has been put on practical estimation of such auctions and derivation of approximations simpler to optimize, such as boosted second-price auctions \cite{Golrezaei2017} or second-price with monopoly reserve price \cite{Roughgarden2016}.

More formally, $\Delta$ being the space of probability distributions over the K bidders, the seller chooses a class of auctions $\cA = \{(\lR_+^K, \bQ, \bM)\}$ consisting in a pair of functions $\bQ : \lR_+^K \to \Delta$ (allocation rule) and $\bM : \lR_+^K \to \lR_+^K$ (payment rule). Here, $Q_i(b)$ is the probability for bidder $i$ to win when bidders submit $b$ and $M_i(b)$ the corresponding expected payment. We denote 
$$
q_i(b_i) = \lE_{H_{-i}}(Q_i(b_i, B_{-i})) ~~~\text{and}~~~ m_i(b_i) = \lE_{H_{-i}}(M_i(b_i, B_{-i}))
$$
As she only observes  bid distributions $H\in\cH(\cF, \cB)$, the seller picks the revenue maximizing auction,
\begin{align*}
    a^*_H = ((\lR_+^K, \bQ_H^*, \bM_H^*)) = \argmax_{a \in \cA} \sum_{i=1}^K \lE_{H_i}\left(m_i^{(a)}(B)\right)
\end{align*}
Under incentive compatibility and individual rationality, the expected payment of bidder $i$ is
$$
\lE_{H_i}\left(m_i(B_i)\right) = m_i(0)+\lE_H\left(\vValue_{B_i}(B_i)Q_i(B)\right),
$$
see, e.g.,  \cite{Myerson81}. The seller does not observe $F$ but only $H$, and we will always assume that she is not strategic: she optimizes her revenue as if she was observing the valuations, i.e., 
\begin{align*}
    a^*_H = \argmax_{a \in \cA} \lE_{H}\sum_{i=1}^K Q_i^{(a)}(B)\vValue_{B_i}(B_i)
\end{align*}

\subsection{Examples of mechanisms, allocations and payment rules.}
We introduce in this section some of the widely-used (and truthful) auctions mechanisms that we will consider.

\textbf{Vickrey-Clarke-Groves (VCG) mechanisms} are second price auctions with non-anonymous reserve prices, i.e., different reserve prices to different bidder. Then $\cA$ is the set of all possible reserve prices. Different allocation rule exist: an item can be allocated to the highest bidder amongst all those that have cleared their reserve price, or to no one if no reserve price is cleared.  We will call it the \textsl{eager} version of Vickrey-Clarke-Groves (VCG) mechanism.

Another allocation rule dictates to allocate the item to the higher bidder, if she has cleared her reserve price, and to no one otherwise. This version of the VCG mechanism will be called \textsl{lazy} (with anonymous reserve prices, eager and lazy versions coincide).

Computations of optimal non-anonymous reserve prices suffer from being NP-hard \cite{FieldGuideToPersonalizedReservePrices2016} and even APX-hard\cite{Roughgarden2016}. However, \cite{Roughgarden2016} also proved that using the monopoly price $\vValue_i^{-1}(0)$ as reserve price in a VCG auction leads to a 2-approximation of the Myerson auction. It also led \citep{FieldGuideToPersonalizedReservePrices2016} to prove that the \emph{lazy} version of VCG admit the monopoly price as optimal reserve price and is a 2-approximation of the VCG with optimal reserve prices.

\medskip

\textbf{Myerson type of auctions} allocate the item to the highest non-negative (assuming there is one) virtual value of the bids $\vValue_{B_i}(b_i)$ and to no one otherwise. The payment of the winner is the smallest winning bid, i.e., $\vValue_{B_i}^{-1}\Big(\max( 0,\{\vValue_{B_i}(b_j)\}) \Big)$. The Myerson auction maximizes the expected seller revenue at the equilibrium, at least if virtual values are increasing. Moreover, in the symmetric case where the distribution of valuations $F_i$ is the same for all bidder, then Myerson auction coincides with the VCG auction with reserve price $\vValue_{B_i}^{-1}(0)$.

A Myerson-type of auctions would follow the same rules except that any non-linear, increasing transformation $\widetilde{\vValue}$ can be used instead of the actual virtual value. In that case, $\cA$  can be either the set of such transformations, or the set of transformations and non-anonymous reserve prices.

\medskip

\textbf{Boosted second price auctions.} In boosted second price auctions \cite{Golrezaei2017}, the seller chooses  parameters $\{(s_j, r_j)\}_{j=1}^n$ and boosts the bid $b_i$ of bidder $i$ by $s_i$ while keeping a reserve price $r_i$. As a consequence, a bid is transformed into a virtual one through the formula $w_i= s_i(b_i-r_i)$. Then $\cA$ is the set of affine transformations of the bids, or subsets thereof. 

Boosted second price auctions actually correspond to Myerson-type auctions if the distributions of valuations belong to the family of Generalized Pareto (GP) distributions. We refer to Section \ref{sec:genParetoAppendix} in the Appendix for more details on the class of GP distributions. Approximating the actual distribution of valuations by a GP, and computing the optimal mechanism w.r.t.\ theses approximations can actually be simpler than  the optimal Myerson auction \cite{Golrezaei2017}.

\medskip

\textbf{Other mechanisms: posted price, first price, etc.} There exist many other auction mechanisms that are more or less equivalent to one of the former under specific assumptions. For instance, in posted price auctions, the seller fixes a price and the item is sold to one bidder (or all of them) whose valuation is higher than the price. Notice that posting price is equivalent to VCG if only one buyer participates in the auctions.

Computing and learning the optimal posted price can also be done in an online fashion and/or using finite number of samples \cite{KleLei03,BluKumRud03,BubDev17}

\medskip

In first price auctions, the item is sold to the highest bidder (or potentially the bidder with the highest virtual bids, or the highest bidder above his reserve price) but with the simple payment rule: the highest bidder simply pays his winning (virtual) bids. Revenue equivalence principle \cite{krishna2009auction} states that these auctions are, at the equilibrium, more or less equivalent, but the reputed truthfulness (hence simple ``optimal strategies'') of second-price auction make these type of auctions less common, for the moment at least \cite{first_price_move}.

\subsection{Strategic Buyer Problem under Seller Revenue Maximization}
In the literature, the study of revenue maximization is done under the  assumption that buyers bid truthfully even if they are aware of the revenue maximization mechanism. This can be motivated in certain settings where  the seller faces repeated auctions, but buyers change from auction to auction (as on EBay).
When the buyers are interacting  repeatedly with the same seller and if they know the mechanism, such as dynamic reserve prices, it is quite easy to exhibit examples proving that incentive compatibility is lost. Yet, state-of-art comes short of providing the bidders with a better strategy than being truthful.
This lack of understanding of the buyers strategy in presence of dynamic reserve prices is arguably one of the major factor for the shift of the market towards first price auctions \cite{first_price_move}.

We are only considering stationary strategies and assuming that auctions are infinitely repeated and undiscounted, hence the pay-off of bidder $i$, whose valuation is $X_i \sim F$ and bids $B_i = \beta_i(X_i) \sim H_i$, is defined by
\begin{align}
\Pi_i(\beta_i) = \lE\left((X_i-\vValue_{B_i}(B_i))Q^{(a_H^*)}_i(B)\right)\; ,
\label{eq:bidder_payoff}
\end{align}
where $B \sim H$ is the vector of bids send by the bidders.

Then, supposing the seller is choosing the revenue maximization auctions within a given class of truthful auctions (e.g. setting different prices, various boosted second price auctions etc...), the question is whether it's still in the interest of the bidders to remain truthful.
In light of the seller's optimizing behavior (and because we consider infinitely undiscounted repeated auctions), player $i$ faces the following optimization problem. Given the distribution of valuations $X \sim F$ and the bidding strategies of other bidders $B_{j} = \beta_j(X_j) \sim H_j$, solve
\begin{equation}\label{eq:bidder_optim}
\max_{\beta_i}\Pi_i(\beta_i)\;,\qquad \text{subject to }\ a^*_H = \argmax_{a \in \cA} \lE_{H}\left(\sum_{i=1}^K Q_i^{(a)}(B)\vValue_{B_i}(B_i)\right)\;.
\end{equation}

Depending on the auction class $\cA$, the bidders may not be able to derive optimal strategies from \eqref{eq:bidder_optim} and it is not  clear whether they always have  incentive to bid untruthfully. Hence, we propose in Section \ref{sec:linear_shading} to study if several auctions considered in revenue optimization are robust to a simple shading scheme in the context of a bidder optimizing \eqref{eq:bidder_optim}. Then, in Section \ref{sec:myerson}, we compute the equilibrium strategy of  \eqref{eq:bidder_optim} for Myerson auctions with symmetric bidders.

\section{Loss of Incentive Compatibility.}
\label{sec:linear_shading}
To get a sense of whether revenue optimization completely breaks the truthfulness, we first study whether several common auctions used for revenue optimization are robust to a very simple modification of strategy: a linear shading of the truthful bid $B_i = \alpha_i X_i$ -- i.e. $\cB = \{x \mapsto \alpha x: \alpha\in[0,1]^K\}$. Then, the relationship between virtual values  and hazard rate simplifies,
\begin{align*}
    \vValue_B(\alpha x) = \alpha \vValue_X(x)~~~~~~~~~\text{and}~~~~~~~~~\lambda_B(\alpha x) = \frac{1}{\alpha} \lambda_X(x)
\end{align*}
We are going to consider the simple setting where only bidder $i$ is being strategic. The assumption of the other bidders being truthful is not a strong restriction as we do not consider the bidders to be symmetric, hence a bidder $j$ bidding truthfully under a distribution of value $F_j$ can also be seen as a bidder following a strategy $\beta_j$ under a distribution of value $F_j \circ \beta_j$.

A mild technical point related to the issue of support might need to be raised: to the left of the support of $X$ we can define the hazard rate as $0$ and hence $\vValue$ is equal to $-\infty$ there. In particular, if $B$ and $X$ don't have the same support. 

\subsection{Myerson Auction}
For simplicity we write $\vValue_{X_i}(x)=\vValue_i(x)$. In this case we know exactly the strategy of the seller, so we can directly re-write the payment of $i$.
$$
Q^{(a^*_H)}_i(B) = \indicator{\alpha \vValue_i(X_i)\geq \max_{j\neq i}(0,\vValue_j(X_j))}
$$
The expected payoff under linear shading assuming that the other bidders fixed their strategy is 
$$
\Pi_i(\alpha)=\Exp{(X_i-\alpha \vValue_i(X_i))\indicator{\alpha \vValue_i(X_i)\geq \max_{j\neq i}(0,\vValue_j(X_j))}}\;.
$$

The following lemma states that bidder $i$ has  incentives to deviate from truthful bidding -- i.e. to choose $\alpha \neq 1$. To support this claim, we just need to prove that $\frac{\partial \Pi_i(\alpha)}{\partial \alpha}\neq 0$ at $\alpha=1$. 
\begin{lemma}\label{lemma:linearShadingInMyerson}
Suppose that $\vValue_i$ is differentiable and that $\vValue_i(x)\leq \vValue'_i(x) x$ on $[\vValue_i^{-1}(0),\infty)$. Then 
$$
\left.\frac{\partial \Pi_i(\alpha)}{\partial \alpha}\right|_{\alpha=1}<0 \text{ at } \alpha=1\;.
$$
In other words, bidder $i$ has an incentive to shade his/her bid in this case. 
\end{lemma}
The conditions of Lemma \ref{lemma:linearShadingInMyerson} are satisfied when $\vValue$ is linear, which happens when $X_i$ has a generalized Pareto distribution. The condition $\vValue(x)\leq \vValue'(x) x$ is also satisfied for $\vValue$ convex which is not of great interest in the current context. 
\begin{proof}
Recall that 
$$
\Pi_i(\alpha)=\Exp{(X_i-\alpha \vValue_{i}(X_i))\indicator{\alpha \vValue_i(X_i)\geq \max_{j\neq i}(0,\vValue_j(X_j))}}\;.
$$
Let us call $Y_i=\max_{j\neq i} \vValue_i^{-1}(\vValue_j(X_j))$ and   
$G(t)=P(t\geq Y_i)$; let  $g$ be its density. Then
$$
\Pi_i(\alpha)=\Exp{(X_i-\alpha \vValue_{i}(X_i)) G(\vValue_{i}^{-1}(\alpha \vValue_{i}(X_i)))\indicator{X_i\geq \vValue_{i}^{-1}(0)}}\;.
$$
Taking the derivative of this quantity with respect to $\alpha$ and splitting the expectation in two, we get
$$
\left.\frac{\partial \Pi_i(\alpha)}{\partial \alpha}\right|_{\alpha=1}=
\underbrace{\Exp{\vValue_i(X_i)\frac{X_i-\vValue_i(X_i)}{\vValue'_i(X_i)} g(X_i) \indicator{X_i\geq \vValue_i^{-1}(0)}}}_{\triangleq ~I}-\underbrace{\Exp{\vValue_i(X_i)G(X_i)\indicator{X_i\geq \vValue_i^{-1}(0)}}}_{\triangleq ~II}
$$
Then, we recognize the expected payment (see \cite{Myerson81}) to get 
$
II=\Exp{m_i(X_i)}\,,
$
and then by definition of the expected payment, 
$$
II =\Exp{\max(\vValue_i^{-1}(0),Y_i)\indicator{X_i\geq Y_i}\indicator{X_i\geq \vValue_i^{-1}(0)}}\,.
$$

We now focus on the first term. We note that $g(X_i)$ is the density of $Y_i$ evaluated at $X_i$. So rewriting the expectation as an integral and using $x-\vValue_i(x)=(1-F_i(x))/f_i(x)$, we get, after using Fubini,
\begin{align*}
I&=\int \frac{\vValue_i(x)}{\vValue'_i(x)} g(x) (1-F(x)) \indicator{x\geq \vValue_i^{-1}(0)}dx=\int \frac{\vValue_i(x)}{\vValue'_i(x)} g(x) \indicator{x\geq \vValue_i^{-1}(0)}\int_x^{\infty} f(t) dt\;,\\
&=\int \int \frac{\vValue_i(x)}{\vValue'_i(x)}\indicator{x\geq \vValue_i^{-1}(0)} \indicator{t\geq x} f(t) g(x) dt dx\;,\\
&=\Exp{\indicator{Y_i\geq \vValue_i^{-1}(0)}\indicator{X_i\geq Y_i} \frac{\vValue_i(Y_i)}{\vValue'_i(Y_i)}}\;.
\end{align*}

We can now compare $I$ and $II$, term by term. We first note that 
$$
\indicator{X_i\geq Y_i}\indicator{X_i\geq \vValue_i^{-1}(0)} \geq \indicator{Y_i\geq \vValue_i^{-1}(0)}\indicator{X_i\geq Y_i}\;.
$$
This inequality is strict on a set of measure non-zero in our setup, so strict inequality passes to expectations.  Under the assumption we made on $\vValue$ we have by definition 
$$
\indicator{Y_i\geq \vValue_i^{-1}(0)}\frac{\vValue_i(Y_i)}{\vValue'_i(Y_i)}\leq \max(\vValue_i^{-1}(0),Y_i)\;.
$$
Then we obtain that $I\leq II$ and in fact $I<II$ when we have strict inequalities in the above two displays on a common set on non-zero measure. 
\end{proof}


Though we do not need symmetry of the bidders' value distribution, we start by a few examples assuming it for concreteness.  We recall that if $F$ is the cdf of $X_i$, $G(x)=F^{n-1}(x)$ in the case where we have $n$ symmetric bidders. 
\\

\textbf{Example of uniform [0,1] distributions:} 
In this case, $\vValue_i(x)=2x-1$ on [0,1] and $\vValue_i^{-1}(0)=1/2$. Also, $G(x)=x^{n-1}$. Then, using for the instance the representation of the derivative of $\Pi_i(\alpha)$ appearing in the proof of Lemma \ref{lemma:linearShadingInMyerson}, we have  
\begin{align*}
\left.\frac{\partial \Pi_i(\alpha)}{\partial \alpha}\right|_{\alpha=1}=
\int_{1/2}^1 \left(x-\frac{1}{2}\right)[(n-1)-x(n+1))]x^{n-2} dx  
    =-\frac{1}{n2^{n+1}}(
    2^n-1)<0\;.
\end{align*}
Hence, each user has an incentive to shade their bid. We note that the derivative goes to 0 as $n\tendsto \infty$ (see also Fig.\ref{fig:uniform_payoff_myerson} right side), which can be interpreted as saying that as the number of users grows, each user has less and less incentive to shade. We can also observe on Fig.\ref{fig:uniform_payoff_myerson} (left side) that the difference between the payoff at optimal shading $\alpha^*$ and the payoff without shading -- $(\Pi(\alpha^*) - \Pi(0))$ -- decreases with $K$.
Indeed, when $K$ grows, the natural level of competition between the bidders makes the revenue optimization mechanisms (e.g. dynamic reserve price) less useful. Logically, being strategic against it in such case does not help much.
For very few bidders, the contrary happens. 
For $K=2$, we even observe that the optimal strategy is to bid with a shading of $\alpha = 0^+$ to force a price close to $0$ while still winning with probability $1/4$ -- when one is beating his reserve and the opponent is not beating his, with the result of almost doubling the payoff. 
\\

\begin{figure}
	\includegraphics[width=0.8\textwidth]{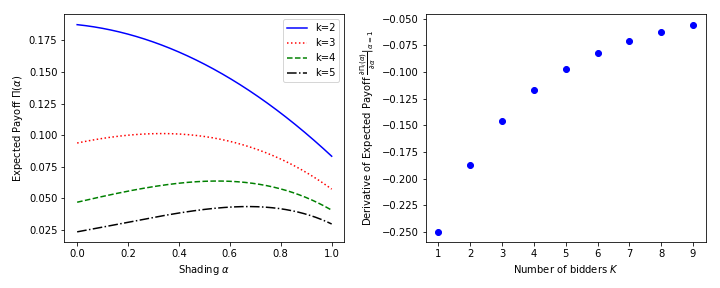}
	\caption{\textbf{Myerson auction : Expected payoff and its derivative for one bidder with linear shading} There are $K$ bidders with values $\Unif[0,1]$, only one of them is strategic. On the left hand side, we present a plot of the expected payoff $\Pi_i(\alpha_i)$ of the strategic bidder for several values of $K$. On the right hand side, we present the derivative $\left.\frac{\partial \Pi_i(\alpha)}{\partial \alpha}\right|_{\alpha=1}$ taken at the truthful bid ($\alpha=1$).}
    \label{fig:uniform_payoff_myerson}	
\end{figure}

In conclusion, we observed that Myerson auction is not immune to deviations for a large number of distributions (generalized Pareto). Even a linear shading can increase substantially the payoff of the bidder at the expense of the seller's revenue. As a corollary, the boosted second price auction incurs the same issue as they are equivalent to Myerson auction for generalized Pareto distributions. Now, we aim at extending these results to other simple of auctions proposed in the context of seller revenue optimization.

\subsection{VCG with Revenue Optimization}
In this section, we  study the robustness of \emph{eager} and \emph{lazy} VCG auction with monopoly price to linear shading of  one bidder. All the computation are very similar for both, differing only by some initial definitions. We denote $a^M_H$ the VCG auction with monopoly reserve price $\vValue_B^{-1}(0)$. Then, the expected payoff of bidder $i$ is 
\begin{align*}
    \Pi_i^M(H) 
    &= \Exp{(X_i-\vValue_{B_i}(B_i))Q^{(a^M_H)}_i(B)}
\end{align*}
Now, considering that all bidders but $i$ are bidding truthfully and bidder $i$ submit bids $B_i = \alpha X_i$, -- i.e. $H_i(b) = F_i(\frac{b}{\alpha})$ -- we obtain the following lemma,
\begin{lemma}\label{lemma:linearShadingInVCG}
For the VCG auction (either \emph{eager} or \emph{lazy}) with monopoly reserve price,
$$
\left.\frac{\partial \Pi^M_i(\alpha)}{\partial \alpha}\right|_{\alpha=1}<0 \text{ at } \alpha=1\;.
$$
In other words, bidder $i$ has an incentive to shade his/her bid even individually. 
\end{lemma}

\begin{proof}
The same proof holds for both \emph{eager} and \emph{lazy}, just by using different definition for $Y_i$. We define $Y_i=\max_{j\neq i} \left\{X_j \indicator{\vValue_j(X_j) \geq 0} \right\}$ for the \emph{eager} one and  $Y_i=\max_{j\neq i} X_j$ for the \emph{lazy} one. Define
$G(t)=P(t\geq Y_i)$ and $g$ its density. Assume that bidder $i$ shades his bids linearly, and  rewrite 
$$
\Pi^M_i(\alpha)=\Exp{(X_i-\alpha \vValue_{i}(X_i)) G(\alpha X_i)\indicator{X_i\geq \vValue_{i}^{-1}(0)}}\;.
$$
Let us compute the partial derivatives of the payoff at $\alpha = 1$.

\begin{align*}
    \left.\frac{\partial \Pi^M_i(\alpha)}{\partial \alpha}\right|_{\alpha=1}
    & = \underbrace{\Exp{X_i(X_i-\vValue_i(X_i)) g(X_i) \indicator{X_i\geq \vValue_i^{-1}(0)}}}_{\triangleq ~I}-\underbrace{\Exp{\vValue_i(X_i)G(X_i)\indicator{X_i\geq \vValue_i^{-1}(0)}}}_{\triangleq ~II}
\end{align*}
As before, we denote that $II = \lE\left(m_i(X_i)\right) = \lE\left(\max(\vValue_{i}^{-1}(0),Y_i) [X_i \geq Y_i] [X_i \geq \vValue_{i}^{-1}(0)] \right)$. Then, similarly as before, we use Fubini argument to obtain
\begin{align*}
I &= \int x_i g(x_i) (1 - F_i(x_i))[\vValue_{i}(x_i) \geq 0] {\rm d}x_i = \int \int_{x_i}^\infty x_i g(x_i) f_i(z)[\vValue_{i}(x_i) \geq 0] {\rm d}z {\rm d}x_i\\
& = \lE\left(Y_i [X_i \geq Y_i] [Y_i \geq \vValue_{i}^{-1}(0)] \right)
\end{align*}
We can now compare $I$ and $II$, term by term. We first note that 
$$
\indicator{X_i\geq Y_i}\indicator{X_i\geq \vValue_i^{-1}(0)} \geq \indicator{Y_i\geq \vValue_i^{-1}(0)}\indicator{X_i\geq Y_i}\;.
$$
This inequality is strict on a set of measure non-zero, so $I < II$.
\end{proof}

\begin{figure}
	\includegraphics[width=0.8\textwidth]{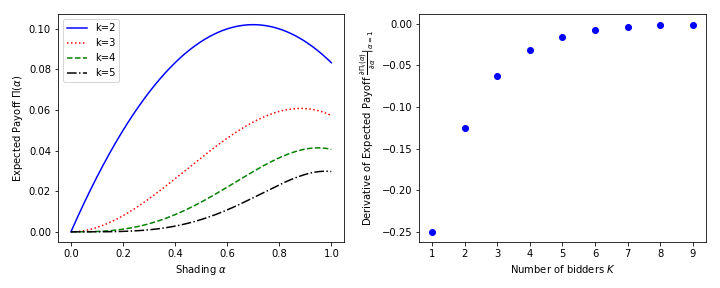}
	\caption{\textbf{VCG \emph{lazy} : Expected payoff and its derivative for one bidder with linear shading} There are $K$ bidders with values $\Unif[0,1]$, only one of them is strategic. On the left hand side, we present a plot of the expected payoff $\Pi_i(\alpha_i)$ of the strategic bidder for several values of $K$. On the right hand side, we present the derivative $\left.\frac{\partial \Pi_i(\alpha)}{\partial \alpha}\right|_{\alpha=1}$ taken at the truthful bid ($\alpha=1$).}
    \label{fig:uniform_payoff_vcg_lazy}	
\end{figure}

We can observe on Figure \ref{fig:uniform_payoff_vcg_lazy} a very similar results for the \emph{lazy} VCG as on Figure \ref{fig:uniform_payoff_myerson} in the Myerson case. There is an incentive to shade the bids, decreasing with the number of bidders. However, the optimal shading is less aggressive than in Myerson case, resulting in a smaller increment of the bidder's payoff. Our intuition is that the class of auctions $\cA$ over which the seller maximizes his revenue is smaller than in Myerson's case, providing less leverage to the bidders to be strategic. Investigating the link between the complexity of $\cA$ (e.g. pseudo-dimension) and the gain for the bidders in being strategic is definitely of interest for future research.

\section{Further results for the Myerson auction}
\label{sec:myerson}
\newcommand{\shadingFunc}{\beta}
\newcommand{\genFunc}{\gamma}
\subsection{Formulation of the problem}\label{subsec:optMyersonPbFormulation}
\subsubsection{Differential equation formulation}
We go back to Myerson auction where bidder 1 shades her bid by bidding $B=\shadingFunc(X_1)$, for a function $\shadingFunc$ to be determined later, instead of her value $X_1$.

Let $\vValue_B$ be the virtual value function associated with the distribution of $B$ (in this section we identify a random variable with its distribution, by a slight abuse of notation). We call $F_Z$ the cumulative distribution function of $Z=\max_{2\leq i \leq K}(0,\vValue_i(Y_i))$, where $Y_i$ is the bid of bidder $i$, and $\vValue_i$ is the virtual value function associated with the distribution of this bid. We assume as before that bidder 1 faces $K-1$ other bidders for a total of $K$ bidders involved in the auction. We call $V_i=\vValue_i(Y_i)$ and $F_{V_i}$ the associated cumulative distribution function. We assume as before that all bidders are independent. In this case, it is clear that $F_Z(x)=\prod_{i=2}^K F_{V_i}(x)\indicator{x\geq 0}$. Of course, this cdf has a jump discontinuity at 0 when $\prod_{i=2}^K F_{V_i}(0)>0$. 

Recall that in the Myerson auction, the expected payment 
of bidder 1 when she bids using $B$ is 
$$
\Exp{\vValue_B(B) \indicator{\vValue_B(B)\geq Z}}\;.
$$
See \cite{krishna2009auction}, p. 67 and \cite{RoughgardenTLA2016} for details. In the Myerson auction, the expected payoff of bidder 1 when she shades through $\shadingFunc$ is therefore
$$
\Pi(\shadingFunc)=\Exp{(X_1-\vValue_B(B))F_Z(\vValue_B(B))}\;.
$$

Suppose that $\shadingFunc \mapsto \shadingFunc_t=\shadingFunc +  t \rho$, where $t>0$ is small and $\rho$ is a function. Let  $B_t=\shadingFunc_t(X_1)$. We denote the corresponding virtual value function by $\vValue_{B_t}$. 

We have (see Lemma \ref{lemma:directionalDerivativeMyerson} in the Appendix) if we differentiate with respect to $t$ and hence take a directional derivative for $\shadingFunc$ in the direction of $\rho$, under mild conditions on $\shadingFunc$, 
\begin{align}\label{eq:directionalDerivativePayoff}
\frac{\partial }{\partial \shadingFunc} \Pi(\shadingFunc)
&=\Exp{\left.\frac{\partial}{\partial t} \vValue_{B_t}(B_t)\right|_{t=0} [(X_1-\vValue_B(B))f_Z(\vValue_B(B))-F_Z(\vValue_B(B))]\indicator{\vValue_B(B)>0}}
\\
&+\left.\frac{\partial}{\partial t} \vValue_{B_t}(B_t) \right|_{t=0,\vValue_B(B)=0}\prod_{i=2}^K F_{V_i}(0) f_{1}(x_{1,\shadingFunc}) x_{1,\shadingFunc}\;,\notag
\end{align}
where $x_{1,\shadingFunc}$ is such that $b=\shadingFunc(x_{1,\shadingFunc})$ and $\vValue_B(b)=0$. We note that $f_Z(t)=0$ and $F_Z(t)=0$ when $t<0$. 

In the work below, we naturally seek a shading function $\shadingFunc$ such that these directional derivatives are equal to 0. We will therefore be interested in particular in functions $\shadingFunc$ such that $[x-\vValue_B(\shadingFunc(x))]f_Z(\vValue_B(\shadingFunc(x)))=F_Z(\vValue_B(\shadingFunc(x)))$, when $\vValue_B(\shadingFunc(x))>0$. The second term in our equation has intuitively to do with the event where the other bidders are discarded for not beating their reserve price. As we will see below, we can sometimes ignore this term, for instance when an equilibrium strategy exists which amounts to canceling the reserve prices. 

Let us first present some intermediary results  to get symmetric equilibrium result among others. 

\subsubsection{Key ODEs and consequences}
As before, we call $X_1$ a non-negative random variable representing the distribution of values of bidder 1 with density $f_1$. For simplicity in defining virtual values we assume that $f_1>0$ on the support\footnote{Taking care of the case where $f_1$ can take the value 0 at a few points introduces artificial technical problems that are not particularly hard to solve but would obscure the flow of our argument.} of $X_1$. 

\begin{lemma}\label{lemma:keyODEs}
Suppose $B=\genFunc(X_1)$, where $\genFunc$ is increasing and differentiable. If $b=\genFunc(x_1)$, we have 
\begin{equation}\label{eq:ODEPhiG}
\vValue_B(b)=\genFunc(x_1)+\genFunc'(x_1)[\vValue_1(x_1)-x_1]\;.
\end{equation}
Furthermore, if for some $x_0$ and a function $h$ we have 
$$
\genFunc_h(x)=\frac{\genFunc_h(x_0)(1-F_1(x_0))-\int_{x_0}^{x} h(u) f_1(u) du}{1-F_1(x)}\;,
$$
then 
\begin{equation}\label{eq:SolnPhiBEqualsh}
\genFunc_h(x_1)+\genFunc_h'(x_1)[\vValue_1(x_1)-x_1]= h(x_1)\;.
\end{equation}
\end{lemma}
\paragraph{\textbf{Interpretation}} Informally, the previous result says that it would be very easy for bidder 1 to shade her bid in such a way that the virtual value of her bid $b$, i.e. $\vValue_B(b)$, be any function $h$ of her value she chooses. 

A simple consequence of Lemma \ref{lemma:keyODEs} is the following result, which pertains directly to non-linear shading strategies.  
\begin{lemma}\label{lemma:makingSureGIncreasing}
Let $h$ be an increasing function. Call $(\mathfrak{l},\mathfrak{u})$ the support of $X_1$. Assume that $f_1>0$ on $(\mathfrak{l},\mathfrak{u})$. ~Then, $\genFunc$ defined as  
\begin{equation}\label{eq:gAsConditionalExpectation}
\genFunc(x)=\frac{\int_x^{\mathfrak{u}} h(t) f_1(t) dt}{1-F_1(x)}=\Exp{h(X_1)|X_1\geq x}
\end{equation}
is increasing and differentiable on the support of $X_1$. 

In particular, if $B=\genFunc(X_1)$, we have, for $b=\genFunc(x_1)$, 
$$
\vValue_B(b)=h(x_1)\;, \forall x_1 \in (\mathfrak{l},\mathfrak{u})\;.
$$
\end{lemma}
The proof of these technical but not difficult lemmas are in the Appendix, Subsection \ref{app:subsec:proofODEs}. 

\subsection{The case of symmetric bidders}\label{subsec:shadingInSymmMyerson}

\begin{theorem}\label{thm:shadingInSymmMyerson}
Consider an auction with $K$ independent and symmetric bidders, having value distribution represented by the random variable $X$. 

In the Myerson auction, a symmetric equilibrium strategy for the bidders is to shade their bids by a function $\shadingFunc_{eq}$ that satisfies 
$$
\shadingFunc_{eq}(x)+\shadingFunc_{eq}'(x)(\vValue_X(x)-x)=\beta^{I}(x)\;,
$$
where $\beta^{I}(x)$ is their symmetric equilibrium first price bid in a first price auction with no reserve price. 

A solution of this equation is 
\begin{align*}
\shadingFunc_{eq,M}(x)&=\Exp{\beta^{I}(X)|X\geq x}\;, \text{ with }\\
b_1&=\shadingFunc_{eq,M}(x_1)=\Exp{\beta^{I}(X)|X\geq x_1} \text{ being the shaded bid of bidder 1}\;.
\end{align*}
With this strategy, the bidders' expected payoffs are the same as what they would get in a first price auction with no reserve price. In particular, it is strictly greater than their expected payoffs had they bid truthfully. 
\end{theorem}
\paragraph{\textbf{Discussion}} The intuition behind this result is quite clear. In the Myerson auction, the expected payoff of any given bidder is the same as that of a first price auction where her bids have been transformed through the use of her virtual value function. We call the corresponding pseudo-bids virtualized bids. Hence, if the bidders can bid in such a way that their ``virtualized" bids are equal to their symmetric equilibrium first price bids, the situation is completely equivalent to a first price auction. And hence their equilibrium strategy in virtualized bid space should be the strategy they use in a standard first price auction with no reserve price. 

Our theorem shows that by adopting such a strategy symmetric bidders can avoid facing a non-zero reserve price. Furthermore, Lemmas \ref{lemma:keyODEs} and \ref{lemma:makingSureGIncreasing} show that it is easy for bidders to shade in such a way that their virtualized bids are equal to any increasing function of their value they choose. Also, this shading is specific to each bidder: the corresponding ordinary differential equations do not involve the other bidders. As such it is also quite easy to implement. 

Nonetheless, the shading is quite counter-intuitive at first, since bidders may end up bidding higher than their value (for instance if their value is 0). Their payments are however made in terms of \emph{virtualized bids}, at least in expectation. And of course, in terms of virtualized bids, nothing is counter-intuitive: everything has been done so that their virtualized bids are equal to their first price bids, which are less than their values.

\begin{proof}
We first investigate properties of $f_Z$ and $F_Z$ in symmetric situations and then verify, as is classical (see \cite{krishna2009auction}, Chapter 2), that our proposed solution is indeed an equilibrium. 

Suppose the bidders shade using the shading function $\shadingFunc$. 
Let us call $h(x)=\vValue_B(\shadingFunc(x))$, when $B_i=\shadingFunc(X_i)$. Note that all $B_i$'s have the same distribution in symmetric equilibrium, which we call $B$. 

In a symmetric equilibrium, everybody will use this shading. Recall that the key relation was 
$$
(x_1-\vValue_{B_1}(\shadingFunc(x_1)))f_Z(\vValue_{B_1}(\shadingFunc(x_1)))-F_Z(\vValue_{B_1}(\shadingFunc(x_1)))=0\;, \text{ for } \vValue_{B_1}(\shadingFunc(x_1)) > 0\;.
$$

$\bullet$ \textbf{Preliminaries 1:} Symmetric situation: property of $F_Z$ and $f_Z$:
If $W=h(X)$, with $h$ increasing, 
$$
F_W(t)=P(W\leq t)=P(h(X)\leq t)=P(X\leq h^{-1}(t))=F_X(h^{-1}(t))\;.
$$
Therefore, 
$$
F_W(t)=F_X(h^{-1}(t))\;,  f_W(t)=\frac{f_X(h^{-1}(t))}{h'(h^{-1}(t))} \; ; \; \; 
F_W(h(x))=F_X(x)\;, f_W(h(x))=\frac{f_X(x)}{h'(x)}\;.
$$
If $X_2,\ldots,X_K$ are independent and are using this strategy, we have for $h(x)>0$, 
$$
F_Z(h(x))=F_X^{K-1}(x) \text{ and } f_Z(h(x))= (K-1) \frac{f_X(x)}{h'(x)} F_X^{K-2}(x)\;.
$$
So the key relation 
$$
(x_1-\vValue_{B_1}(\shadingFunc(x_1)))f_Z(\vValue_{B_1}(\shadingFunc(x_1)))-F_Z(\vValue_{B_1}(\shadingFunc(x_1)))=0\;, \text{ for } \vValue_{B_1}(\shadingFunc(x_1)) > 0\;.
$$
can now be re-written as 
$$
(x-h(x))(K-1)\frac{f_X(x)}{h'(x)}=F_X(x)\;, 
\text{ or } 
\boxed{(K-1)(x-h(x)) f_X(x)-h'(x)F_X(x)=0}\;, h(x)\geq 0\\
$$
$\bullet$ \textbf{Preliminaries 2:} Connection with first price auctions\\
\noindent We recognize here the equation for $h$ defining the shading strategy in a first price auction with no reserve price (see \citet{krishna2009auction}, Chapter 2). We can also solve this equation in a very simple way. Indeed, the most general solution of this differential equation is just  
$$
h=h_0 F^{-(K-1)} \text{ with } h_0'(x)=\frac{K-1} xf_X(x) F_x^{K-2}\;.
$$
Taking 
$$
h(x)=\frac{\int_0^x yg(y) dy}{G(x)}\;, \text{ where } G(x)=F^{K-1}(x)\;, \text{ and } g=G'\;,
$$
we find a solution that is increasing, with $h(x)>0$ for $x>0$. Of course, $h$ can be reinterpreted as 
$$
h(x)=\Exp{Y_1|Y_1<x}=\beta^{I}(x), 
$$ where $Y_1$ is the maximum bid of bidder 1's $(K-1)$ competitors in the auction. This of course is nothing but a symmetric equilibrium first price bid in the symmetric case with no reserve price, see \cite{krishna2009auction}, p. 15. 

$\bullet$ \textbf{Verification argument}
Equipped with the results we derived above, the last step of the proof is just a verification argument. 
We note that $\beta^{I}(x)$ is an increasing function of $x$. Furthermore, if $\genFunc$ is a solution of 
$$
\genFunc(x)+\genFunc'(x)(\vValue_X(x)-x)=\beta^{I}(x)\;,
$$
such that 
$$
\genFunc(x)=\frac{\int_x^{\mathfrak{u}}\beta^{I}(t)f(t) dt}{1-F(x)}=\Exp{\beta^{I}(X)|X\geq x}\;,
$$
we have seen in Lemma \ref{lemma:makingSureGIncreasing} that $\genFunc$ is increasing under our assumptions when $h(x)=\beta^I(x)$ is increasing as is the case here. Also $\beta^{I}(x)>0$ for $x>0$ and $\beta^{I}(x)=0$ if $x=0$.  ($\mathfrak{u}$ is the right end point of the support of $X$.) 

We conclude that our function $\shadingFunc_{eq,M}=\Exp{\beta^{I}(X)|X\geq x}$ is increasing. 

We need to verify that the problem we are dealing with is regular, so that the payoff of the Myerson auction is indeed what we announced and in particular, no ironing is necessary (see \citet{Myerson81} or \citet{Toikka2011}) . Almost by definition, we have, if $b=\shadingFunc_{eq,M}(x_1)$,
$$
\vValue_B(b)=\vValue_B(\shadingFunc_{eq,M}(x_1))=\beta^{I}(x_1)=\beta^{I}(\shadingFunc_{eq,M}^{-1}(b))\;.
$$
Since $\shadingFunc_{eq,M}$ is increasing, so is $\shadingFunc_{eq,M}^{-1}$; and since $\beta^{I}$ is increasing, so is $\beta^{I}\circ\shadingFunc^{-1}_{eq,M}$. We conclude that $\vValue_B$ is increasing and so the design problem for the seller receiving $B$ is regular. 

Now as we are in the symmetric case, if our adversaries use this strategy, bidder 1's opposition in virtual value space is $Z=\max(0,\vValue_{B_i}(\shadingFunc_{eq,M}(X_i)))=\max_{2\leq i \leq K}(0,\beta^{I}(X_i))=\max_{2\leq i \leq K}(\beta^{I}(X_i))$. For this last equality we have used the well-known and easy to verify fact that $\beta^{I}(X_i)\geq 0$ with probability 1. So the distribution of $Z$ is continuous, it has no discontinuity at 0. In particular, for the functional derivative of our payoff we have, in the notation of Subsection \ref{subsec:optMyersonPbFormulation}
$$
\frac{\partial}{\partial \shadingFunc }\Pi(\shadingFunc_{eq,M})=\Exp{\left.\frac{\partial}{\partial t} \vValue_{B_t}(B_t)\right|_{t=0} [(X_1-\vValue_B(B))f_Z(\vValue_B(B))-F_Z(\vValue_B(B))]}\;.
$$
Now we have done everything so that with $b=\shadingFunc_{eq,M}(x_1)$ as above (and $B=\shadingFunc_{eq,M}(X_1))$)
$$
[(x_1-\vValue_B(\shadingFunc_{eq,M}(x_1)))f_Z(\vValue_B(\shadingFunc_{eq,M}(x_1)))-F_Z(\vValue_B(\shadingFunc_{eq,M}(x_1)))]=0 \text{ for all } x_1 \text{ in the support of } X_1\;.
$$

Of course the same reasoning applies to the other bidders. To finish the equilibrium proof, we note that the expected payoff of the bidder shading her bid as described above is the same as what she would get in a first price auction with no reserve price. Using the revenue equivalence principle (see \cite{krishna2009auction}, Chapter 3), this is also what she would get in a second price auction with no reserve price. Of course, the standard Myerson auction where bidders are symmetric and bid truthfully amounts to performing a second price auction with reserve price set at the monopoly price $\vValue_X^{-1}(0)$. The expected payoff of the bidders for this latter auction is clearly strictly less than that in the second price auction with no reserve price. We conclude that the expected payoff of the symmetric bidders using the strategy described above is strictly greater than what they would have gotten had they bid truthfully. 
\end{proof}


\subsection{The case of one strategic bidder}
The case of one strategic bidder is also of interest, and motivated by the different nature of the various bidders involved in online advertising auctions. Recent work on boosted second price auctions \cite{Golrezaei2017} was in part motivated by the desire to account for this diversity, between for instance what these authors call brand bidders and retargeting bidders, and also to simplify the implementation of the Myerson auction. Recall that \cite{Golrezaei2017} propose to effectively linearize the virtual value of the bidders before applying a Myerson-type approach on these ``linearly-virtualized" bids. 

In this context, distributions with exactly linear virtual value play a particular role. It is easy to show that those distributions are Generalized Pareto (GP) distributions. We refer to  \cite{Balseiro2015MultiStage} and \cite{AllouahBesbes2017} for their use in the auction context and for pointing out their remarkable simplicity in terms of virtual value computations.

When doing explicit computations or focusing on boosted second price auctions we will naturally also make use of this family of distributions.

\subsubsection{General formulation}
The problem faced by bidder 1 in the Myerson auction has not changed. She seeks to shade her bid through a mapping $\shadingFunc$, i.e. bid $B=\shadingFunc(X_1)$ so as to maximize 
$$
\Pi(\shadingFunc)=\Exp{[X_1-\vValue_B(B)] F_Z(\vValue_B(B))}\;.
$$
Nonetheless, two aspects of the problem will now differ from our earlier work: on the one hand $F_Z$, which represents the bid distributions she is facing is now ``fixed", i.e. unaffected by $\shadingFunc$, because the other bidders are non-strategic; on the other hand, the shading function $\shadingFunc$ will on occasion be considered to be part of a parametric family. In this case, what was before a directional derivative will then be a simple gradient. 

With this in mind, and with a slight overloading of notations (since in the parametric case $\partial/\partial t$ below is just a gradient), we recall the key relationship 
\begin{align*}
\frac{\partial }{\partial \shadingFunc} \Pi(\shadingFunc)
&=\Exp{\left.\frac{\partial}{\partial t} \vValue_{B_t}(B_t)\right|_{t=0} [(X_1-\vValue_B(B))f_Z(\vValue_B(B))-F_Z(\vValue_B(B))]\indicator{\vValue_B(B)>0}}
\\
&+\left.\frac{\partial}{\partial t} \vValue_{B_t}(B_t) \right|_{t=0,\vValue_B(B)=0}\prod_{i=2}^K F_{V_i}(0) f_{1}(x_{1,\shadingFunc}) x_{1,\shadingFunc}\;,\notag
\end{align*}
where $x_{1,\shadingFunc}$ is such that $b=\shadingFunc(x_{1,\shadingFunc})$ and $\vValue_B(b)=0$.
We will be keenly interested in shading functions $\shadingFunc$ such that $$
(x_1-\vValue_B(\shadingFunc(x_1)))f_Z(\vValue_B(\shadingFunc(x_1)))=F_Z(\vValue_B(\shadingFunc(x_1)))\;, \text{ when } \vValue_B(\shadingFunc(x_1))>0\;. 
$$
Indeed, for those $\shadingFunc$'s, the expectation in our differential will be 0. Hence, computing the differential will be relatively simple and in particular will give us reasonable guesses for $\shadingFunc$ and descent directions, even if it does not always give us directly an optimal shading strategy. Furthermore, when $K$ is large, the second term fades out, as the probability that no other bidder clear their reserve prices becomes very small. 

If we proceed formally, and call, for $x>0$, $h(x)=(\id+F_Z/f_Z)^{-1}(x)$ (temporarily assuming that this - possibly generalized - functional inverse can be made sense of), we see that solving the previous equation amounts to solving
$$
\vValue_B(\shadingFunc(x_1))=(\id+F_Z/f_Z)^{-1}(x_1)=h(x_1)\;.
$$
Lemma \ref{lemma:keyODEs} can of course be brought to bear on this problem. We note that we will be primarily interested in solutions of this equation for $x_1$'s such that $\vValue_B(\shadingFunc(x_1))>0$. 

\subsubsection{Explicit computations in the case of Generalized Pareto families}
It is clear that now we need to understand $F_Z$, $F_Z/f_Z$ and related quantities to make progress. By definition, if $X$ is Generalized Pareto (GP) with parameters $(\mu,\sigma,\xi)$, we have, when $\xi<0$ and $t\in [\mu,\mu-\sigma/\xi]$, 
$$
P(X\geq t)=(1+\xi (t-\mu)/\sigma)^{-1/\xi}\;
$$
and otherwise $P(X\geq t)=\exp(-(t-\mu)/\sigma)$ if $\xi=0$. See Appendix \ref{subsec:remindersGP} for further details.  In GP families, the virtual value has the form $\vValue(t)=c_\vValue(t-r^*)$, where $r^*$ is the monopoly price and $c_\vValue=1-\xi$.

\begin{lemma}\label{lemma:variousGPComps}
Suppose $Y$ has a Generalized Pareto distribution. Call $F_Y$ the cdf of $Y$ and $f_Y$ its density. 

If $V=\vValue_Y(Y)$, where $\vValue_Y$ is the virtual value of $Y$, we have $F_V(t)=F_Y(\vValue_Y^{-1}(t))$ and 
$$
\frac{F_V(t)}{f_V(t)}=c_\vValue \frac{F_Y(\vValue_Y^{-1}(t))}{f_Y(\vValue_Y^{-1}(t))}\;,
$$
where $c_{\vValue_Y}=\vValue_Y'(t)$. If $r^*_Y$ is the monopoly price associated with $Y$, we more specifically have 
$$
\vValue_Y^{-1}(t) = \frac{t}{c_{\vValue_Y}}+r^*_Y\;, \text{ and }\frac{F_V(t)}{f_V(t)}=c_{\vValue_Y} \frac{F_Y(t/c_{\vValue_Y}+r^*_Y)}{f_Y(t/c_{\vValue_Y}+r^*_Y)}\;.
$$
\end{lemma}
In case $F_Y$ has finite support, we naturally restrict $t$ to values such that $x=\vValue_Y^{-1}(t)$ is just that $f_Y(x)>0$. 
\begin{proof}
$F_V$ is just the cumulative distribution function of $\vValue_Y(Y)$, where $\vValue_Y$ is the virtual value of $Y$. Hence, since in GP families $\vValue_Y$ is increasing, 
$$
F_V(t)=P(V\leq t)=P(\vValue_Y(Y)\leq t)=F_Y(\vValue_Y^{-1}(t))\;.
$$
In particular, 
$$
f_V(t)=\frac{f_Y(\vValue_Y^{-1}(t))}{\vValue_Y'(\vValue_Y^{-1}(t))}\;.
$$
In Generalized Pareto families, $\vValue_Y$ is linear, so that $\vValue_Y'$ is a constant, because $\vValue_Y(t)=c_{\vValue_Y}(t-r^*_Y)$, where $r^*_Y$ is the monopoly price. The first result follows immediately. Noticing that $\vValue_Y^{-1}(x)=x/c_{\vValue_Y}+r^*_Y$ gives the second result.  
\end{proof}
The previous lemma yields the following useful corollary. 
\begin{corollary}
Suppose $K\geq 2$, $Y_2,\ldots,Y_K$ are independent, identically distributed, with Generalized Pareto distribution. Call $\vValue_Y$ their virtual value function and $Z=\max_{2\leq i \leq K}(0,\vValue_Y(Y_i))$. Then, if $F_Z$ is the cumulative distribution function of $Z$, we have 
$$
\frac{F_Z(t)}{f_Z(t)}=\frac{c_{\vValue_Y}}{K-1} \frac{F_Y(t/c_{\vValue_Y}+r^*_Y)}{f_Y(t/c_{\vValue_Y}+r^*_Y)}\;, \text{ for } t>0 \text{ and such that } f_Y(t/c_{\vValue_Y}+r^*_Y)>0\;.
$$
\end{corollary}
\subsubsection{An example: uniform non-strategic bidders}
In this subsection we assume that bidder 1 is facing $K-1$ other bidders, with values $Y_i$'s that are i.i.d $Unif[0,1]$. In this case, $c_{\vValue_Y}=2$ and $r^*_Y=1/2$ so $F_Z(t)=\min(1,[(t+1)/2]^{K-1})$ for $t>0$.   
We recall that the $Unif[0,1]$ distribution is GP(0,1,-1). Bidder 1 is strategic whereas bidders 2 to $K$ are not and bid truthfully. 

\begin{lemma}[Shading against $(K-1)$ uniform bidders]\label{lemma:shadingAgainstUnifBidders}
Suppose that $X_1$ has a density that is positive on its support. We assume for simplicity that $X_1$ is bounded by $(K+1)/(K-1)$. Let $\eps>0$ be chosen by bidder 1 arbitrarily close to 0. 
Let us call 
$$
h^{(\eps)}_K(x)=
\begin{cases}
\frac{K-1}{K} \frac{\eps}{1+\eps} x & \text{ if } x \in [0,(1+\eps)/(K-1)) \;,\\
\frac{K-1}{K}\left(x-\frac{1}{K-1}\right) & \text{ if } x\geq (1+\eps)/(K-1)\;.
\end{cases}
$$
A near-optimal shading strategy is for bidder 1 to shade her value through
$$
\shadingFunc^{(\eps)}_1(x_1)=\Exp{h^{(\eps)}_K(X_1)|X_1\geq x_1}\;.
$$
As $\eps$ goes to $0^+$, this strategy approaches the optimum. 

If the support of $X_1$ is within $(1/(K-1),(K+1)/(K-1))$, then $\eps$ can be taken equal to 0.
\end{lemma}

\begin{proof}
If we call $h(x)=\psi_B(\shadingFunc(x))$ we can in this case write bidder 1's expected payoff directly using the results of the previous subsection: 
$$
\Pi(\shadingFunc)=\int_{x:h(x)>0} (x-h(x))\min\left(1,\frac{[h(x)+1]^{K-1}}{2^{K-1}}\right) f_1(x) dx\;.
$$
In light of the fact we want to maximize this integral as a function of $h$, with the requirement that $h> 0$, it is natural to study the function $f_x(c)=(x-c)[c+1]^{K-1}$. 

If we call $h_K(x)=\argmax_{c\geq 0} f_x(c)$, we can split the problem into two cases. If $x>1/(K-1)$, $h_K(x)=\frac{K-1}{K}\left(x-\frac{1}{K-1}\right)$. Note that with our assumption that $x\leq (K+1)/(K-1)$, $h_K(x)\leq 1$. For $x<1/(K-1)$, the function $f_x(\cdot)$ is decreasing for $c\geq 0$. Hence, 
$$
h_K(x)=\argmax_{c\geq 0}(x-c)[c+1]^{K-1}=
\begin{cases}
0 & \text{ if } x\leq 1/(K-1)\;,\\
\frac{K-1}{K}\left(x-\frac{1}{K-1}\right) & \text{ if } x>1/(K-1) \;.
\end{cases}
$$
Recall that for Lemma \ref{lemma:makingSureGIncreasing} to apply, we need to integrate an increasing function and $h_K$ is not increasing on $(0,\infty)$. 

However, bidder 1 can use the following $\eps$-approximation strategy: let us call 
$$
h^{(\eps)}_K(x)=
\begin{cases}
\frac{K-1}{K} \frac{\eps}{1+\eps} x & \text{ if } x \in [0,(1+\eps)/(K-1)) \;,\\
\frac{K-1}{K}\left(x-\frac{1}{K-1}\right) & \text{ if } x\geq (1+\eps)/(K-1)\;.
\end{cases}
$$
Notice that $\sup_x |h^{(\eps)}_K(x)-h_K(x)|<\eps/K$. In light of Lemmas \ref{lemma:keyODEs} and \ref{lemma:makingSureGIncreasing}, the corresponding  function
$$
\shadingFunc^{(\eps)}_1(x_1)=\Exp{h^{(\eps)}_K(X_1)|X_1\geq x_1}
$$
is increasing and will then guarantee an expected payoff that is nearly optimal since it will be 
$$
\Pi(\shadingFunc^{(\eps)})=\frac{1}{2^{K-1}}\int (x-h_K^{(\eps)}(x))[h_K^{\eps)}(x)+1]^{K-1} f_1(x) dx\;.
$$
It can be made arbitrarily close to optimal by decreasing $\eps$. Using $\eps>0$ guarantees that the virtualized bid $h^{(\eps)}_K(x)$ is always strictly positive and hence effectively sends the monopoly price for bidder 1 to 0. (Even if $\shadingFunc^{(0)}$ is increasing, a potential problem might occur if the virtualized bid is exactly zero. Using $\shadingFunc^{(\eps)}$ with $\eps=0^+$ solves that problem.)

If $X_1$ is supported on a subset of $[1/(K-1),(K+1)/(K-1))$, taking $\eps=0$ is possible and optimal. 
\end{proof}
The assumption that $X_1\leq (K+1)/(K-1)$ can easily be dispensed of as the proof makes clear : one simply needs to look for the argmax of another function. Our main example follows and does not require taking care of this minor technical problem. 
\paragraph{\textbf{Case where bidder 1 has value distribution $\Unif[0,1]$}} We first note that $X_1\leq 1\leq (K+1)/(K-1)$, so Lemma \ref{lemma:shadingAgainstUnifBidders} applies as-is. We therefore have 
$$
\shadingFunc^{(\eps)}_1(x)=
\begin{cases}
\frac{K-1}{K}[\frac{1}{2}(1+x)-\frac{1}{K-1}] & \text{ if } 1\geq x\geq x_{\eps}=\frac{1+\eps}{K-1}\;,\\
\frac{K-1}{K} \frac{1}{1-x}\left(\frac{\eps}{1+\eps} \frac{1}{2} (x_\eps^2-x^2)+\shadingFunc^{(\eps)}_1(x_\eps) (1-x_\eps)\right) & \text{ if } x<\frac{1+\eps}{K-1}\;.
\end{cases}
$$
Taking $\eps$ to 0 yields
$$
\shadingFunc_1(x)=
\begin{cases}
\frac{K-1}{K}[\frac{1}{2}(1+x)-\frac{1}{K-1}] & \text{ if } x\geq \frac{1}{K-1}\;,\\
\frac{1}{1-x} \frac{(K-2)^2}{2(K-1)K} & \text{ if } x<\frac{1}{K-1}\;.
\end{cases}
$$
See Figure \ref{fig:uniform_bid_profiles} for a plot of $\shadingFunc_1$ and comparison to other possible shading strategies.  

Similar computations can be carried out if $X_1$ has another GP distribution. For those distributions, the shading beyond $1/(K-1)$ is also affine in the value of bidder 1, $x_1$. (See Appendix \ref{app:tediousCompGPs} for relevant details.) Interestingly, it is easy to verify that affine transformations of GP random variables are GP. However, if the support of $X_1$ includes part of $(0,1/(K-1))$, $\shadingFunc_1(X_1)$ will not have a GP distribution in general.

\begin{figure}
	\includegraphics[width=0.8\textwidth]{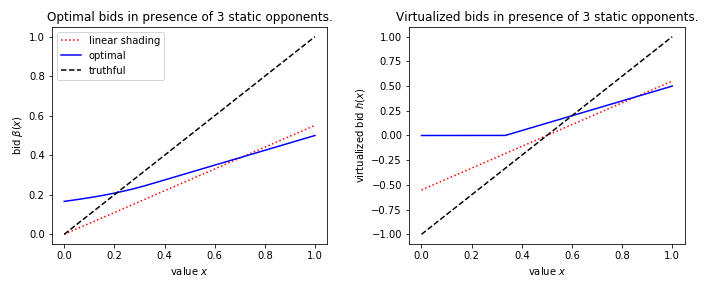}
	\caption{\textbf{Myerson auction: Bids and virtualized bids with one strategic bidder} There are K=4 bidders, only one of them is strategic. On the left hand side, we present a plot of the bids sent to the seller. ``Linear shading" corresponding to a bid $\shadingFunc_\alpha(x)=\alpha x$, where $x$ is the value of bidder 1; here $\alpha$ is chosen numerically to maximize that buyer's payoff - see Lemma \ref{lemma:linearShadingInMyerson}. ``Optimal" corresponds to the strategy described in Lemma \ref{lemma:shadingAgainstUnifBidders}, with $\eps=0^+$. On the right hand side (RHS), we present the virtualized bids, i.e. the value taken by the associated virtual value functions evaluated at the bids sent to the seller. This corresponds on average to what the buyer is paying in the Myerson auction, when those virtualized bids clear 0. We can interpret the RHS figure as showing that for both optimal and linear shading,  a strategic buyer end up winning more often and paying less (conditional on the fact that she won) than if she had been truthful; this explains why her average payoff is higher than with truthful bidding.}
\label{fig:uniform_bid_profiles}	
\end{figure}

\subsubsection{Boosted second price auctions: shading through GP families }
Motivated by boosted second price auctions \cite{Golrezaei2017}, we study the question of shading one's bid in a way that guarantees that the distribution of bids sent to the seller is in the GP family. In this case, boosted second price auctions turn into Myerson auctions.   

It seems natural that a strategic bidder facing a boosted second price auction would shade her bid by sending bids with distribution in the Generalized Pareto family\footnote{we discuss here only deterministic shadings into the GP families; stochastic shadings are of course possible by relying on copulas. See \cite{NelsenCopulas2006} for a review of those tools.} since this would facilitate the computation of her expected payoff and also limit the uncertainty appearing otherwise because of the implementation of the linearization of her virtual value function (see \cite{Golrezaei2017} for description of that procedure) for non-GP random variables.

If bidder 1 has value $X_1$, a random variable whose cumulative distribution function is $F_1$, which we assume to be a continuous function. Then $U_1=1-F_1(X_1)$ has a uniform distribution. Let us call 
\begin{align*}
\mathfrak{p}&=(\mu,\sigma,\xi)\;, \text{ and } \\
X_{\mathfrak{p}}&=\frac{\sigma}{\xi}\left[(1-F_1(X_1))^{-\xi}-1\right]+\mu\;, \;
x_{\mathfrak{p}}=\frac{\sigma}{\xi}\left[(1-F_1(x_1))^{-\xi}-1\right]+\mu\;.
\end{align*}
It is clear (see Appendix \ref{subsec:remindersGP}) that $X_{\mathfrak{p}}$ has GP($\mu,\sigma,\xi$) distribution; furthermore $x_{\mathfrak{p}}$ is an increasing function of $x_1$. Bidder 1's payoff in a boosted second price auction is
$$
\Pi(\mathfrak{p})=\Exp{(X_1-\vValue_\mathfrak{p}(X_\mathfrak{p}))F_Z(\vValue_\mathfrak{p}(X_\mathfrak{p}))}
\text{ and she seeks }
\mathfrak{p}^*=\argmax_{\mathfrak{p}}\Pi(\mathfrak{p})\;.
$$
This is a 3-parameter optimization problem. It is in general non-convex but numerical optimization methods could nonetheless be used by a strategic bidder to increase her expected payoff. See Subsection \ref{subsec:bspaGradComps} for further details. 
\subsection{Beyond the Myerson auction and future work}
The key element in getting the main results of Subsection \ref{subsec:shadingInSymmMyerson} was a representation of the bidder's expected payoff as a function of her value, the shading function she used, and the distribution of the competition she faced in terms of virtual bids. 

It seems possible to obtain such a representation for other types of auctions, such as certain second price auctions with personalized reserve prices (see \cite{FieldGuideToPersonalizedReservePrices2016}).  The technique we have used in Sections \ref{sec:linear_shading} and \ref{sec:myerson} is simply to ``integrate out" the competition so as to be able to formulate a functional optimization problem in terms of a single bidder. 
We leave the details of this and related problems to future work. 

Finally, another potentially more practical strategy would be to shade bids infinitesimally and to adjust this infinitesimal shading to reactions of the other bidders and the seller. While connected to the main theme of our paper, this approach raises a number of questions that are quite different in nature from those that form the core of this work and we plan to present our work on infinitesimal shading in another piece of work.

\section{Conclusion}
\label{sec:conclusion}
We have shown in this paper that revenue optimization on the seller side turns a number of classical auction formats that are widely purported to be truthful into auctions that are not incentive compatible. This is especially relevant in the context of repeated auctions and Internet advertising. 

Bidders' shading strategies can be as simple as linearly shading bids. This already results in increased payoffs for them at both ends of the complexity spectrum for auctions, from certain second price auctions to the Myerson auction. For this latter auction, we also find that it is a symmetric equilibrium for bidders to shade their bids in a non-linear way resulting in the same expected payoff for them as in a first price auction with no reserve price. The tools we develop in the paper also allow us to exhibit non-linear shading strategies when only one buyer is strategic while the others bid truthfully. 

This suggests that, counter-intuitively, adopting revenue-optimizing auctions may not bring more revenue to the seller. It does however turn otherwise simple means of exchange of goods into quite opaque ones. Our work gives theoretical grounding to the oft-heard practitioners' call for return to simple auctions, in particular in the world of online advertising auctions. 

\section{Aknowledgements}
N. El Karoui gratefully acknowledges support from grant NSF DMS 1510172. V. Perchet has benefited from the support of the FMJH Program Gaspard Monge in optimization and operations research (supported in part by EDF), from the Labex LMH and from the CNRS through the PEPS program.

\bibliographystyle{ACM-Reference-Format}
\bibliography{main}

\newpage
\appendix
\appendix
	\begin{center}
	\textbf{\textsc{APPENDIX}}
	\end{center}
	\vspace{1cm}
\newcommand{\partialDerVValue}{\left.\frac{\partial }{\partial t}\vValue_{B_t}(b_t)\right|_{t=0}}

\section{Directional derivatives and the Myerson auction}
\subsection{Directional derivatives}
Suppose $\shadingFunc$ is an increasing and differentiable function, $X$ is a random variable with positive density on its support and $B=\shadingFunc(X)$. We call $\vValue_B$ the virtual value function of $B$. $F_Z$ is a cumulative distribution function of the form $F_Z(x)=\prod_{i=2}^K F_{V_i}(x) \indicator{x\geq 0}\triangleq \Gamma_K(x) \indicator{x\geq 0}$. We assume that $F_{V_i}$'s are differentiable (i.e. $V_i$'s have a density) and therefore so is $\Gamma_K(x)$. We further assume that $\Gamma_K$ is differentiable and call its derivative $\gamma_K$ (or $f_Z$ when its argument is positive). 

We will also assume below that $\shadingFunc$ is such that $\vValue_B$ is increasing; see Lemmas \ref{lemma:keyODEs} and \ref{lemma:makingSureGIncreasing} to see how this requirement can be enforced.  

We are interested in 
$$
\Pi(\shadingFunc)=\Exp{(X-\vValue_B(B))F_Z(\vValue_B(B))}\;
$$ 
and its directional derivative. 
\begin{lemma}\label{lemma:directionalDerivativeMyerson}
Suppose that $\shadingFunc$ has the properties mentioned above and is such that $\vValue_B$ is increasing.

Call $x_\shadingFunc$ the point such that $\vValue_B(\shadingFunc(x_\shadingFunc))=0$. 

Let us call $\shadingFunc_t=\shadingFunc+t \rho$ where $\rho$ is another function differentiable function. $\shadingFunc_t$ is also assumed to be increasing and differentiable, at least for $t$ infinitesimally small. 

Then, the directional derivative of $\Pi(\shadingFunc)$ in the direction of $\rho$ is 
\begin{align*}
\frac{\partial }{\partial \shadingFunc}\Pi(\shadingFunc) &=
\Exp{\left.\frac{\partial }{\partial t}\vValue_{B_t}(B_t)\right|_{t=0}\left\{[X-\vValue_B(B)] f_Z(\vValue_B(B))-F_Z(\vValue_B(B))\right\}\indicator{\psi_B(B)>0}}\\
&+ \left.\frac{\partial }{\partial t}\vValue_{B_t}(b_t)\right|_{t=0,\value_B(b)=0}\prod_{i=2}^K F_{V_i}(0) f(x_\shadingFunc) x_\shadingFunc
\end{align*}

\end{lemma}
\begin{proof}
We call $\shadingFunc_t=\shadingFunc+t \rho$ where $\rho$ is another function. We assume that $\shadingFunc_t$ is also increasing and differentiable. Let us call $b_t=\shadingFunc_t(x)$ and $B_t=\shadingFunc_t(X)$. Note that using Equation \eqref{eq:ODEPhiG}, we have 
$$
\frac{\vValue_{B_t}(b_t)-\vValue_{B}(b)}{t}=\rho(x)+\rho'(x)(\vValue_X(x)-x)\;,
$$
Hence, as a function of $\shadingFunc$, $\vValue_{B}(b)$ admits a well defined directional derivative, which we call $\left.\frac{\partial }{\partial t}\vValue_{B_t}(b_t)\right|_{t=0}$. Furthermore it is very easy to express in terms of $\rho$, $\vValue_X$ and $x$.  To compute the directional derivative of $\Pi$, we can therefore essentially differentiate under the expectation sign. 

We note that $F_Z$ is differentiable in the sense of distributions and we have, if $\delta_0$ denotes a Dirac point mass at 0,  
$$
F'_Z(x)=\left\{
\begin{array}{ll}
\gamma_K(x) &\text{ if } x>0\;,\\
\delta_0 \prod_{i=2}^K F_{V_i}(0) & \text{ if } x=0\;,\\
0 & \text{ if } x<0\;.
\end{array}
\right.
$$
If we rewrite $\Pi(\shadingFunc)$ as an integral we have 
$$
\Pi(\shadingFunc)=\int f_X(x) (x-\vValue_B(\shadingFunc(x))) F_Z(\vValue_B(\shadingFunc(x))) dx\;.
$$
Calling $x_\shadingFunc$ the point such that $\vValue_B(\shadingFunc(x_\shadingFunc))=0$, we have after taking the directional derivative under the integral 
\begin{align*}
\frac{\partial }{\partial \shadingFunc}\Pi(\shadingFunc) &=
\int  \left.\frac{\partial }{\partial t}\vValue_{B_t}(b_t(x))\right|_{t=0} f_X(x) \left\{[x-\vValue_B(\shadingFunc(x))] \gamma_K(\vValue_B(\shadingFunc(x)))-F_Z(\vValue_B(\shadingFunc(x))\}\right) dx\\
&+ \left.\frac{\partial }{\partial t}\vValue_{B_t}(b_t)\right|_{t=0,\value_B(b)=0}\prod_{i=2}^K F_{V_i}(0) f(x_\shadingFunc) x_\shadingFunc\;.
\end{align*}
Reinterpreting the first integral as an expectation concludes the proof. 
\end{proof}

\subsection{Proofs of ODE lemmas}\label{app:subsec:proofODEs}

We start by proving Lemma \ref{lemma:keyODEs} and turn to Lemma \ref{lemma:makingSureGIncreasing} afterwards.
\begin{proof}[Proof of Lemma \ref{lemma:keyODEs}]
If $B=\genFunc(X_1)$, $F_B(x)=P(B\leq x)=P(\genFunc(X_1)\leq x)=F_1(\genFunc^{-1}(x))$, since $\genFunc$ is increasing. Hence, for the density of $B$, we have
$$
f_B(x)=f_1(\genFunc^{-1}(x))\frac{1}{\genFunc'(\genFunc^{-1}(x))}\;.
$$
Therefore, 
\begin{align*}
\vValue_B(x)&=x-\frac{1-F_B(x)}{f_B(x)}
=x-\genFunc'(\genFunc^{-1}(x))\frac{1-F_1(\genFunc^{-1}(x))}{f_1(\genFunc^{-1}(x))}\\
&=x+\genFunc'(\genFunc^{-1}(x))\left[\vValue_1(\genFunc^{-1}(x))-\genFunc^{-1}(x)\right]\;,\text{ and }\\
\vValue_B(b)&=\vValue_B(\genFunc(x_1))=\genFunc(x_1)+\genFunc'(x_1)\left[\vValue_1(x_1)-x_1\right] \text{ as } b=\genFunc(x_1)\;.
\end{align*}
This proves Equation \eqref{eq:ODEPhiG}. 
\newcommand{\genFuncTwo}{\eta}
For the second part of the lemma, we recall that $\vValue_1(x_1)=(1-F_1(x_1))/f_1(x_1)$ (see \citet{krishna2009auction}, Chap. 5, p. 68). Hence, for any function $\genFuncTwo$, we have, when $f_1(x_1)>0$, 
\begin{align*}
\genFuncTwo(x_1)+\genFuncTwo'(x_1)\left[\vValue_1(x_1)-x_1\right]&=\genFuncTwo(x_1)-\genFuncTwo'(x_1)\frac{1-F_1(x_1)}{f_1(x_1)}\;
=\frac{\genFuncTwo(x_1) f_1(x_1)+ \genFuncTwo'(x_1) (F_1(x_1)-1)}{f_1(x_1)}\\
&=\frac{(\genFuncTwo(F_1-1))'(x_1)}{f_1(x_1)}
\end{align*}
It is now easy to verify by direct computation that if 
$$
\genFunc_h(x)=\frac{\genFunc_h(x_0)(1-F_1(x_0))-\int_{x_0}^{x} h(u) f_1(u) du}{1-F_1(x)}\;,
\text{ we have }
\genFunc_h(x)-\genFunc_h'(x)\frac{1-F_1(x)}{f_1(x)}=h(x)\;.
$$
\end{proof}
We now turn to proving Lemma \ref{lemma:makingSureGIncreasing}.
\begin{proof}[Proof of \ref{lemma:makingSureGIncreasing}]
In light of Lemma \ref{lemma:keyODEs}, the only question we have to settle is whether $\genFunc$ is increasing and differentiable. If that is the case Equation \eqref{eq:ODEPhiG} applies and then Equation \eqref{eq:SolnPhiBEqualsh} can be interpreted as stating that $\vValue_B(b)=h(x_1)$. 

$\genFunc$ is clearly differentiable and after differentiating it, it is clear that showing that $\genFunc$ is increasing amounts to showing that for all $x\in (\mathfrak{l},\mathfrak{u})$, 
$$
\frac{\int_x^{\mathfrak{u}}h(t) f_1(t) dt }{1-F_1(x)}=\Exp{h(X_1)|X_1\geq x}>h(x)=\inf_{t\in [x,\mathfrak{u})} h(t)\;.
$$
The last equality is true because $h$ is assumed to be increasing. 

For the sake of completeness, we prove by elementary means this trivial inequality. Under our assumptions on $h$, we can find $\delta>0$ and $x\leq x_\delta<\mathfrak{u}$ such that if $y\geq x_\delta$, $h(y)\geq h(x)+\delta$. (For instance, if $\mathfrak{u}$ is finite, take $x_\delta=(x+\mathfrak{u})/2$ and the corresponding $\delta$. If $\mathfrak{u}$ is infinite, take $x_\delta=2x$.) Hence, $\int_{x}^\mathfrak{u} h(u) f_1(u) du\geq h(x) (F_1(x_\delta)-F_1(x))+(h(x)+\delta) (1-F_1(x_\delta))=h(x) (1-F_1(x))+\delta(1-F_1(x_\delta))$. Since $F_1(x_\delta)<1$ because of the definition of $x_\delta$ and $f_1>0$, the result is shown. 

This shows that $\genFunc$ is increasing and the Lemma is shown. 
\end{proof}

\section{Generalized Pareto distributions: reminders and details about boosted second price auctions}\label{sec:genParetoAppendix}

\newcommand{\GP}{\text{GP}}
\subsection{Reminders: definitions and basic results}\label{subsec:remindersGP}
If $X$ is Generalized Pareto (GP) with parameters $(\mu,\sigma>0,\xi)$, we have, when $\xi<0$, 
$$
P(X\geq t)=(1+\xi (t-\mu)/\sigma)^{-1/\xi}\;.
$$
and otherwise $P(X\geq t)=\exp(-(t-\mu)/\sigma)$ if $\xi=0$. The support is $[\mu,\mu-\sigma/\xi]$. 

2-parameter GP families are also often considered; in this case $\mu=0$ and we denote the corresponding random variable as $GP(\sigma,\xi)$. For the GP$(\sigma,\xi)$, we have the following simple results:
\begin{itemize}
\item $$\Exp{X}=\frac{\sigma}{1-\xi}\;.$$
\item The virtual value is 
$$
\vValue(x)=(1-\xi)\left(x-\Exp{X}\right)\;.
$$
\item In particular, the monopoly price $\vValue^{-1}(0)$, which we denote by $r^*$ is such that 
$$
r^*=\Exp{X}\;.
$$
\item We have the stochastic representation 
$$
X=\frac{\sigma}{\xi}(U^{-\xi}-1)\;, \text{ where } U\sim \Unif[0,1]\;.
$$
\end{itemize}


For GP($\mu,\sigma,\xi$), the monopoly price is of the form $r^*=\sigma/(1-\xi)-\xi\mu/(1-\xi)$ and we have
$$
\vValue(x)=(1-\xi)(x-r^*)=(1-\xi)(x-\Exp{X})-\mu\;.
$$
Of course, stochastically, if $X_\mu\sim GP(\mu,\sigma,\xi)$ and $X_0\sim GP(\sigma,\xi)$, 
$$
X_\mu=\mu+X_0\;.
$$

\subsection{An integral involving GP distributions}\label{app:tediousCompGPs}
Suppose $F_1(x)=1-(1+\xi x/\sigma)^{-1/\xi}$, i.e. $X$ is $GP(\sigma,\xi)$. For its density we have $f_1(x)=\frac{1}{\sigma}(1+\xi x/\sigma)^{-1/\xi-1}$ on $[0,-\sigma/\xi]$. 
A tedious computation shows that 
$$
\int_x^{-\sigma/\xi} u f_1(u) du =\frac{\sigma}{\xi}\left[\frac{1}{1-\xi}\left(1+\frac{\xi x}{\sigma}\right)-1\right](1-F_1(x))\;.
$$
The result is easy to verify by differentiation.

The conclusion is that 
\begin{equation}\label{eq:meanGPs}
r(x)=\frac{\int_x^{-\sigma/\xi} u f_1(u) du }{1-F_1(x)}=\Exp{X|X\geq x}=\frac{x+\sigma}{1-\xi}\;.
\end{equation}
Note that, as is obvious from the integral definition,  $$
r(x)\geq x\;,
$$
because 
$$
r(x)=\frac{x+\sigma}{1-\xi}=x+\frac{\sigma}{1-\xi}(1+\xi x/\sigma)\geq x\;, \text{ since } x\leq -\frac{\xi}{\sigma}\;. 
$$
\subsection{Boosted second price auctions: gradient computations}
\label{subsec:bspaGradComps}


If bidder 1 has value $X_1$, a random variable whose cumulative distribution function is $F_1$, which we assume to be a continuous function. Then $U_1=1-F_1(X_1)$ has a uniform distribution. Let us call 
\begin{align*}
\mathfrak{p}&=(\mu,\sigma,\xi)\;, \text{ and } \\
X_{\mathfrak{p}}&=\frac{\sigma}{\xi}\left[(1-F_1(X_1))^{-\xi}-1\right]+\mu\;, \;
x_{\mathfrak{p}}=\frac{\sigma}{\xi}\left[(1-F_1(x_1))^{-\xi}-1\right]+\mu\;.
\end{align*}
It is clear that $X_{\mathfrak{p}}$ has GP($\mu,\sigma,\xi$) distribution; furthermore $x_{\mathfrak{p}}$ is an increasing function of $x_1$. Bidder 1's payoff in a boosted second price auction is, with the same notations we had above, 
$$
\Pi(\mathfrak{p})=\Exp{(X_1-\vValue_\mathfrak{p}(X_\mathfrak{p}))F_Z(\vValue_\mathfrak{p}(X_\mathfrak{p}))}
\text{ and she seeks }
\mathfrak{p}^*=\argmax_{\mathfrak{p}}\Pi(\mathfrak{p})\;.
$$
Here we focus on gradient computations for use in numerical optimization. 
If $u_1=(1-F_1(x_1))$, it is easy to verify that 
$$
\vValue_{\mathfrak{p}}(x_{\mathfrak{p}})=(1-\xi)x_{\mathfrak{p}}-\sigma +\xi \mu =\frac{1-\xi}{\xi}\sigma u_1^{-\xi}-\frac{\sigma}{\xi}+\mu\;.
$$
In particular,  
$$
\nabla_{\mathfrak{p}}\vValue_{\mathfrak{p}}(x_{\mathfrak{p}})=
\begin{bmatrix}
1  \\
\frac{1-\xi}{\xi}u_1^{-\xi} -\frac{1}{\xi}\\
\frac{\sigma}{\xi^2}[1-u_1^{-\xi}+(1-\xi) \ln(u_1^{-\xi}) u_1^{-\xi}]
\end{bmatrix}\;.
$$
\newcommand{\rvValueBSPA}{\vValue_{\mathfrak{p}}(X_{\mathfrak{p}})}
\newcommand{\dvValueBSPA}{\vValue_{\mathfrak{p}}(x_{\mathfrak{p}})}
It is also the case that when $\vValue_{\mathfrak{p}}(x_{\mathfrak{p}})=0$, $u_1^{-\xi}=(1-\mu\xi/\sigma)/(1-\xi)$. We conclude, using computations similar to those leading to Equation \eqref{eq:directionalDerivativePayoff} that 
\begin{align*}
\nabla_{\mathfrak{p}}\Pi(\mathfrak{p})&=
\Exp{\nabla_{\mathfrak{p}}\rvValueBSPA[(X_1-\rvValueBSPA)f_Z(\rvValueBSPA)-F_Z(\rvValueBSPA)]\indicator{\rvValueBSPA>0}}\\
&+\begin{bmatrix}
1\\
-\frac{\mu}{\sigma}\\
\frac{\sigma}{\xi^2}\left\{\frac{\xi}{1-\xi} (\frac{\mu}{\sigma}-1)+(1-\frac{\mu\xi}{\sigma})\ln[u_1^{-\xi}]\right\}
\end{bmatrix}
\prod_{i=2}^K F_{V_i}(0) f_1(x_1) x_1\;,
\end{align*}
where $u_1=1-F_1(x_1)$ and $u_1^{-\xi}=(1-\mu\xi/\sigma)/(1-\xi)$. This result could be used to implement first order optimization methods for optimization in the setting of boosted second price auctions. 

In case $F_Z$ is known, and as could be probably surmised in boosted second price auctions corresponds to GP bids, Lemma \ref{lemma:variousGPComps} would prove useful. In particular, in this situation, a reasonable first approach would be to neglect the second term of the previous equation and find $\mathfrak{p}$ such that $(x_1-\dvValueBSPA)f_Z(\dvValueBSPA)-F_Z(\dvValueBSPA)]\indicator{\dvValueBSPA>0}=0$. This of course amounts to solving a differential equation of the same kind we have solved several times in this paper and so we do not pursue this avenue of research in detail in this paper.

\end{document}